\documentclass[letterpaper, 10 pt, conference]{ieeeconf}  
\IEEEoverridecommandlockouts             
\usepackage{bm}
\usepackage{enumerate}
\usepackage{commath}
\usepackage{graphicx}
\graphicspath{{Pictures/}}
\usepackage{amsmath}
\usepackage{subcaption}
\usepackage{breqn}
\usepackage{cite}
\usepackage[table]{xcolor}
\usepackage{booktabs}
\usepackage{empheq}
\usepackage{amsfonts}
\usepackage{amssymb}

\usepackage{amsthm}
\usepackage{hyperref}
\usepackage{xcolor}
\usepackage{mathrsfs}

\usepackage{accents}
\usepackage{thmtools}
\usepackage{thm-restate}
\usepackage{float}
\usepackage{comment}
\usepackage[normalem]{ulem}
\usepackage{algorithm}
\usepackage{algpseudocode}

\theoremstyle{plain}

\newtheorem{definition}{Definition}
\newtheorem{assumption}{Assumption}

\newtheorem{proposition}{Proposition}

\newtheorem*{problem*}{Problem}
\newtheorem*{theorem*}{Theorem}
\newtheorem{assumption*}{Assumption}

\declaretheorem[name=Theorem]{thm}

\theoremstyle{definition}
\newtheorem{example}{Example}

\definecolor{amber}{rgb}{1.0, 0.3, 0.0}

\newcommand{\redtext}[1]{{\color{red}#1}}

\newcommand{\myvar}[1]{\bm{#1}}
\newcommand{\myvarfrak}[1]{\bm{\mathfrak{#1}}}

\newcommand{\myvardot}[1]{\dot{\myvar{#1}}}

\newcommand{\myset}[1]{\mathcal{#1}} %% for this paper

\begin{document}

\title{
Compatibility checking of multiple control barrier functions \\ for input constrained systems
}

\author{Xiao Tan, and Dimos V. Dimarogonas % <-this % stops a space
% \thanks{*This work was not supported by any organization}% <-this % stops a space
\thanks{ This work was supported in part by Swedish Research Council (VR), in part by  Swedish Foundation for Strategic Research (SSF) COIN Project, in
part by  ERC CoG LEAFHOUND, in part by EU CANOPIES Project, and in part by  Knut and Alice Wallenberg Foundation (KAW). The authors are with the School of EECS, Royal Institute of Technology (KTH), 100 44 Stockholm, Sweden (Email: 
        {\tt\small xiaotan, dimos@kth.se}).}
}

\maketitle
\thispagestyle{plain}
\pagestyle{plain}

%%%%%%%%%%%%%%%%%%%%%%%%%%%%%%%%%%%%%%%%%%%%%%%%%%%%%%%%%%%%%%%%%%%%%%%%%%%%%%%%

\begin{abstract}
  State and input constraints are ubiquitous in control system design. One recently developed tool to deal with these constraints is control barrier functions (CBF) which transform state constraints into conditions  in the input space. CBF-based controller design thus incorporates both the CBF conditions and input constraints in a quadratic program. However, the CBF-based controller is  well-defined only if  the CBF conditions are compatible. In the case of perturbed systems, robust compatibility is of relevance. In this work, we propose an algorithmic solution to verify or falsify the (robust) compatibility of given CBFs \textit{a priori}. Leveraging the Lipschitz properties of the CBF conditions, a grid sampling and refinement method  with theoretical analysis and guarantees is proposed.

\end{abstract}

%%%%%%%%%%%%%%%%%%%%%%%%%%%%%%%%%%%%%%%%%%%%%%%%%%%%%%%%%%%%%%%%%%%%%%%%%%%%%%%%

\section{Introduction}

Control design for dynamical systems with input and state constraints is an omnipresent problem in engineering and has been extensively investigated over the last few decades. Among all the investigated control methods, such as model predictive control  (MPC)\cite{mayne2000constrained}, reference governor \cite{bemporad1998reference}, barrier Lyapunov functions (BLF) \cite{Tee2009}, and prescribed performance control (PPC) \cite{bechlioulis2008robust}, the so-called control barrier functions (CBF) has revived recently\cite{wieland2007constructive,Xu2015a} and gained increasing popularity in the control and robotic community. The latter three methods relate a system state and possibly a time state with a real number.  Unlike its counterparts as in BLF and PPC, CBF is negative if the system state is in the unsafe/undesired regions. By enforcing an inequality constraint on the system input, which is referred to as the CBF condition, the system trajectory is guaranteed to stay within a given safe region for all time. A CBF-based controller is thus formed as a quadratic program (QP)\cite{Xu2015a} with the CBF conditions and actuator limits as the constraints while trying to minimize its differences to a pre-designed task-satisfying controller.  A moderate magnitude  control signal is obtained when the system state is close to the boundary of the safety set, making it applicable even in the presence of noise. Compared to MPC schemes, the CBF formulation only requires to solve a small size quadratic program online and thus is more suitable for embedded systems thanks to today's increasing computational power.

% Why multiple CBFs 
Another nice property of the CBF formulation is its modular design feature. In the case that multiple state constraints are present for the system, i.e., the safe region is the intersection of all these state constraints, we can add multiple CBF conditions to the QP formulation each of which corresponds to one state constraint. Under the core assumption that the CBF-induced QP is always feasible, or equivalently, the CBF conditions are compatible,  the satisfaction of all the state constraints is guaranteed for all time.

However, the compatibility of multiple CBFs is in general difficult to check. The problem is even more challenging if input limits are also present. In \cite{xu2018constrained}, sufficient and necessary conditions on CBF compatibility are discussed for SISO systems without input constraints. In the recent work \cite{cortez2021robust}, the authors propose a mixed-initiative control formulation that satisfies the input bound explicitly and enforces CBF conditions only in a neighborhood of the safety boundary. Under an assumption that the neighborhood around the boundary of safety region specified by each CBF should not overlap, the mixed-initiative formulation is applicable in the presence of multiple CBFs. In \cite{notomista2021safety},  a navigation problem is considered and a CBF-based QP for the image of obstacles in the ``ball world" is employed. The QP feasibility is guaranteed thanks to the special structure of the problem, yet input constraints are missing.

More relevant to our work is the sum-of-square (SoS) approach in \cite{prajna2007framework}  for verifying worst-case and stochastic system safety by constructing a barrier certificate. A recent work\cite{clark2021verification} extends the SoS technique to verify if a candidate function is indeed a CBF by checking whether a polynomial control input  exists and satisfies the CBF condition in the safe region. In \cite{axton2022on}, a SoS-based compatibility verification scheme for multiple CBFs is proposed. However, these results are only applicable to polynomial control systems.  Moreover, the feasibility result is also subject to the order of the polynomials, and a failure to find such polynomial inputs provides no falsification guarantee.

In this paper, we consider the compatibility checking problem when multiple control barrier functions are present for input constrained systems. We aim to give a verification or falsification on the compatibility of multiple CBFs prior to their online implementation. In that respect, a grid sampling and refinement method is proposed leveraging the Lipschitz properties of the CBF conditions. We show that 1) the proposed algorithm will output the exact compatibility result if it terminates, 2) the algorithm is guaranteed to terminate in finite steps if the multiple CBFs are robustly compatible, and 3) the upper bound of the robustness level can be obtained if a lower bound of the lattice size is incorporated.

\section{Preliminaries}

\textit{Notation}: The operator $\nabla:C^1(\mathbb{R}^n) \to \mathbb{R}^n$ is defined as the gradient $\frac{\partial}{\partial x}$ of a scalar-valued differentiable function with respect to $\myvar{x}$. The Lie derivatives of a function $h(\myvar{x})$ for the system $\myvardot{x} = \myvarfrak{f}(\myvar{x}) + \myvarfrak{g}(\myvar{x}) \myvar{u}$ are denoted by $L_{\mathfrak{f}} h = \nabla h^\top \myvarfrak{f}(\myvar{x}) \in \mathbb{R}$ and $L_{\mathfrak{g}} h = \nabla h^\top \myvarfrak{g}(\myvar{x}) \in \mathbb{R}^{1\times m}$, respectively. 
The interior and boundary of a set $\myset{A}$ are denoted $\text{Int}(\myset{A})$ and $\partial \myset{A}$, respectively.
A continuous function $\alpha:[0,a) \to [0,\infty)$ for $a \in \mathbb{R}_{>0}$ is a \textit{class $\mathcal{K}$ function} if it is strictly increasing and $\alpha(0) = 0$ \cite{Khalil2002}. A continuous function $\alpha:(-b,a) \to (-\infty,\infty)$ for $a,b \in \mathbb{R}_{>0}$ is an \textit{extended class $\mathcal{K}$ function} if it is strictly increasing and $\alpha(0) = 0$.
Vector inequalities are to be interpreted element-wise. $\myvar{0}, \myvar{1}$ refer to vectors of proper dimensions with all entries to be $0$ or $1$, respectively.

Consider the nonlinear control affine system
\begin{equation} \label{eq:nonlinear_dyn}
    \myvardot{x} = \myvarfrak{f}(\myvar{x}) + \myvarfrak{g}(\myvar{x}) \myvar{u},
\end{equation}
where the state $\myvar{x} \in \mathbb{R}^n$, and the control input $ \myvar{u} \in \mathbb{U} \subset \mathbb{R}^m$. Assume that the vector fields $\myvarfrak{f}(\myvar{x})$ and $\myvarfrak{g}(\myvar{x})$ are locally Lipschitz functions in $\myvar{x}$. A set $\myset{A} \subset \mathbb{R}^n$ is called \textit{forward invariant}, if for any initial condition $ \myvar{x}_0 \in \myset{A}  $, the system solution $\myvar{x}(t,\myvar{x}_0)  \in \myset{A}$ for all $t$ in the maximal time interval of existence.

Consider the safety set $\myset{C}$ defined as an intersection of superlevel sets of continuously differentiable functions $h_i:\mathbb{R}^n \to \mathbb{R}, i \in \mathcal{I} = \{1,2,...,N\}$:
\begin{equation} \label{eq:set_c}
    \myset{C} = \{ \myvar{x}\in \mathbb{R}^n: h_i(\myvar{x})\ge 0, i\in \mathcal{I} \}.
\end{equation}

\begin{definition}[Compatible CBFs] \label{def:compatible_cbf}
The functions $h_i(\myvar{x}), i \in \mathcal{I}$ are \textit{compatible} control barrier functions (CBF) for \eqref{eq:nonlinear_dyn} if there exists an open set $\myset{D}\supseteq \myset{C}$ and locally Lipschitz extended class $\mathcal{K}$ functions $\alpha_i$ such that, $ \forall \myvar{x}\in \myset{D},   \forall i \in \mathcal{I},$
\begin{equation} \label{eq:cbf}
    \exists \myvar{u} \in \mathbb{U}, \  L_{\myvarfrak{f}}h_i(\myvar{x}) + L_{\myvarfrak{g}}h_i(\myvar{x})\myvar{u} + \alpha_i(h_i(\myvar{x})) \ge 0.
\end{equation}
\end{definition}

For given differentiable functions $h_i$ and extended class $\mathcal{K}$ functions $\alpha_i$, $i\in \mathcal{I}$, define
\begin{multline} \label{eq:k_x}
        \myset{K}(\myvar{x}) = \{\myvar{u}\in \mathbb{U}: L_{\myvarfrak{f}}h_i(\myvar{x})   +  L_{\myvarfrak{g}}h_i(\myvar{x})\myvar{u} \\ + \alpha_i(h_i(\myvar{x})) \ge 0, \forall i\in \mathcal{I} \}.
\end{multline}
Then $h_i(\myvar{x}), i\in \mathcal{I},$ being compatible CBFs is equivalent to that $\myset{K}(\myvar{x})\neq \emptyset, \forall \myvar{x}\in \myset{D}$. 

\begin{proposition} \label{prop:foward_invariance}
If $h_i(\myvar{x}), i \in \mathcal{I}$, are compatible CBFs, then any locally Lipschitz continuous feedback control law $\myvar{u}(\myvar{x})\in \myset{K}(\myvar{x})$ renders the safe set $\myset{C}$ forward invariant.
\end{proposition}
\begin{proof}
    This is evident from the Brezis version of Nagumo's Theorem at $\partial \myset{C}$. Please check \cite[Theorem 4]{redheffer1972theorems} for details.
\end{proof}

In the case that the system is subject to disturbances/uncertainty, a relevant concept is that of \textit{robustly compatible} CBFs.
\begin{definition}[Robustly compatible CBFs]
The functions $h_i(\myvar{x}), i \in \mathcal{I}$ are \textit{robustly compatible} control barrier functions with robustness level $\eta>0$ for system \eqref{eq:nonlinear_dyn}, if there exists an open set $\myset{D}\supseteq \myset{C}$ and locally Lipschitz extended class $\mathcal{K}$ functions $\alpha_i$ such that, $ \forall \myvar{x}\in \myset{D},   \forall i \in \mathcal{I},$
\begin{equation} \label{eq:robust_cbf}
    \exists \myvar{u} \in \mathbb{U}, \  L_{\myvarfrak{f}}h_i(\myvar{x}) + L_{\myvarfrak{g}}h_i(\myvar{x})\myvar{u} + \alpha_i(h_i(\myvar{x}))\ge \eta.
\end{equation}
\end{definition}

The condition in \eqref{eq:robust_cbf} is stricter compared to the condition in \eqref{eq:cbf}. If \eqref{eq:robust_cbf} holds, then the safety set $\myset{C}$ can be rendered forward invariant for the perturbed system $\dot{\myvar{x}} = \mathfrak{f}(\myvar{x}) + \mathfrak{g}(\myvar{x}) \myvar{u} + \mathfrak{p}(\myvar{x})\myvar{\omega}$ as long as $|  L_{\mathfrak{p}\myvar{\omega}}h_i(\myvar{x})|\leq \eta$, $\forall i\in \mathcal{I},\forall \myvar{x}\in \myset{D}$. The analysis follows similarly as in \cite{jankovic2018robust}\cite[Remark 3]{xiao2021high}.

A CBF-based safety controller $\myvar{u}:\myset{D}\to \mathbb{R}^m$ is given in the following form 
\begin{equation} \label{eq:cbf_controller}
\begin{aligned}
        \myvar{u}(\myvar{x}) & = \arg \min_{\myvar{v}} \| \myvar{v} - \myvar{u}_{nom}(\myvar{x})\| \\
        \textup{s.t. }  & \myvar{v} \in \myset{K}(\myvar{x}),
\end{aligned}
\end{equation}
where $\myvar{u}_{nom}$ is a nominal controller focusing on task completion. For example, $\myvar{u}_{nom}$ can be designed for state stabilization, reference tracking or can be given directly by a human user. One core problem for the CBF-based controller formulation in \eqref{eq:cbf_controller} is the compatibility, i.e., $\myset{K}(\myvar{x})\neq \emptyset, \forall \myvar{x}\in \myset{D}$. When external disturbances are present, robustly compatibility is a desired property for practical implementation.

\begin{comment}
\redtext{This is for a later analysis. }
The following three problems  need to be answered properly:
\begin{enumerate}
    \item[Q1] Given $h_i(\myvar{x}), \alpha_i(\cdot)$ and control input set $\mathbb{U}$, verify or falsify the hypothesis $h_i(\myvar{x})$s are compatible;
    \item[Q2] Given $h_i(\myvar{x})$ and control input set $\mathbb{U}$, how to choose $ \alpha_i(\cdot)$s that fulfill the assumption $h_i(\myvar{x})$s are compatible;
    \item[Q3] If no such $ \alpha_i(\cdot)$s  exist that fulfill the previous assumption, which safety constraint should be set ``relaxed".
\end{enumerate}
\end{comment}

In this work, we propose an algorithmic solution to  verify or falsify the hypothesis that $h_i(\myvar{x}), i\in \mathcal{I}$ are (robustly) compatible. The compatibility verification algorithm only needs to be executed once and offline, before applying the CBF-based safety controller \eqref{eq:cbf_controller} online. For notational brevity, given $h_i(\myvar{x}), \alpha_i(\cdot)$ and the control system in \eqref{eq:nonlinear_dyn}, we denote 
\begin{equation*}
    A(\myvar{x}):=\begin{pmatrix}
L_{\myvarfrak{g}}h_1(\myvar{x}) \\
L_{\myvarfrak{g}}h_2(\myvar{x}) \\
... \\
L_{\myvarfrak{g}}h_N(\myvar{x})
\end{pmatrix}, b(\myvar{x}):=\begin{pmatrix}
L_{\myvarfrak{f}}h_1(\myvar{x}) + \alpha_1(h_1(\myvar{x})) \\
L_{\myvarfrak{f}}h_2(\myvar{x}) + \alpha_2(h_2(\myvar{x})) \\
... \\
L_{\myvarfrak{f}}h_N(\myvar{x}) + \alpha_N(h_N(\myvar{x}))
\end{pmatrix}.
\end{equation*}
The problem is thus to verify whether
\begin{equation} \label{eq:feasibility problem}
    \sup_{\myvar{u}\in \mathbb{U}} A(\myvar{x})\myvar{u} + b(\myvar{x}) \ge \myvar{0}, \forall \myvar{x}\in \myset{D}.
\end{equation}
for compatibility, and whether
\begin{equation} \label{eq:robust feasibility problem}
    \sup_{\myvar{u}\in \mathbb{U}} A(\myvar{x})\myvar{u} + b(\myvar{x}) \ge \eta\myvar{1}, \forall \myvar{x}\in \myset{D}.
\end{equation}
for robust compatibility with robustness level $\eta>0$.

To simplify our analysis though without jeopardizing the generality, we assume the following: 
\begin{assumption} \label{ass:compact C}
The safe set $\myset{C}$ is compact.
\end{assumption}

\begin{assumption} \label{ass:convex U}
The input set $\mathbb{U}$ is convex.
\end{assumption}
\noindent Under Assumption \ref{ass:convex U} and in view of \eqref{eq:k_x}, we know that $\myset{K}(\myvar{x})$ is either empty or convex, for any $\myvar{x}\in \myset{D}$.

\section{Proposed solutions}
\subsection{Grid sampling algorithm using $n$-cubes}
We recall some basic notions for approximating a compact set in $\mathbb{R}^n$ using $n$-cubes. Let $\myset{S}\subset \mathbb{R}^n$ be a compact set, and $\{ \myvar{e}_i, i = 1,2,...,n\}$ the canonical basis of $\mathbb{R}^n$.  For $\myvar{x}\in \mathbb{R}^n, r>0$, define
\begin{multline}
\label{eq:lattice}
    P_{\textup{lattice}}(\myvar{x},r) = \{ \myvar{y} \in \mathbb{R}^n:   \myvar{y} = \myvar{x} + 
    \sum_{i\in \{1,2,...,n\}}  a_i r \myvar{e}_i,\\
    \forall a_i \in  \mathbb{N}, i = \{1,2,...,n\} \},
\end{multline}
\begin{multline}
 \label{eq:n cubes}
    B(\myvar{x},r) = \{ \myvar{y}\in \mathbb{R}^{n}: \myvar{y} = \myvar{x} +  \sum_{i\in \{1,2,...,n\}} k_i r \myvar{e}_i, \\
    \forall k_i\in [-1/2,1/2], i\in \{1,2,...,n\} \}.
\end{multline}
Here $P_{\textup{lattice}}$ denotes a set of points that forms a regular lattice with size $r$ in $\mathbb{R}^n$ and $\myvar{x}\in P_{\textup{lattice}}$; $B(\myvar{x},r)$ denotes a $n$-cube in $\mathbb{R}^n$ centered at $\myvar{x}$ with size $r$.

Now we propose the following grid sampling algorithm. First we calculate the range limit $\rho_{\myvar{e}_i}^{\min}$ and $\rho_{\myvar{e}_i}^{\max}, i = 1,2, ..., n$ of the set $\myset{S}$ (Line 1 of Algorithm \ref{alg:grid sampling}). Since $\myset{S}$ is compact, $\myset{S} $ is a subset of the hyperrectangle $ [\rho_{\myvar{e}_1}^{\min}, \rho_{\myvar{e}_1}^{\max}] \times [\rho_{\myvar{e}_2}^{\min}, \rho_{\myvar{e}_2}^{\max}]\times ... \times [\rho_{\myvar{e}_n}^{\min}, \rho_{\myvar{e}_n}^{\max}]$ (Line 2).  Then we construct a regular lattice $P_{\textup{lattice}}$ around the center point of the hyperrectangle with size $r$. In Line 3, we obtain a set $P_{\textup{cand}}$ by intersecting $P_{\textup{lattice}}$ with the inflated hyperrectangle $ [\rho_{\myvar{e}_1}^{\min}-r/2,\rho_{\myvar{e}_1}^{\max}+r/2]\times [\rho_{\myvar{e}_2}^{\min}-r/2,\rho_{\myvar{e}_2}^{\max}+r/2] \times ... \times [\rho_{\myvar{e}_n}^{\min}-r/2,\rho_{\myvar{e}_n}^{\max}+r/2] $. We then collect all the points $\myvar{p}$ in $P_{\textup{cand}}$ around which the $n$-cube with size $r$ intersects with the set $\myset{S}$ (Line 4). The algorithm returns $G$ as a Cartesian product of $P$ and the singleton $\{r\}$.

\begin{algorithm} 
\caption{\texttt{GridSampling}} \label{alg:grid sampling}
\begin{algorithmic}[1]
\Require Compact set $ \myset{S}\subset\mathbb{R}^n$, lattice size $r$
\State   Calculate $ \rho_{\myvar{e}_i}^{\min} = \min_{\myvar{x}\in \myset{S}} \myvar{e}_i^\top \myvar{x},\rho_{\myvar{e}_i}^{\max} = \max_{\myvar{x}\in \myset{S}} \myvar{e}_i^\top \myvar{x} $ for $i \in \{1,2,...,n\}$.
\State  Construct a regular lattice $P_{\textup{lattice}}$ around $(\frac{\rho_{\myvar{e}_1}^{\min}+\rho_{\myvar{e}_1}^{\max} }{2}, \frac{\rho_{\myvar{e}_2}^{\min}+\rho_{\myvar{e}_2}^{\max} }{2}, ..., \frac{\rho_{\myvar{e}_n}^{\min}+\rho_{\myvar{e}_n}^{\max} }{2} )$ with size $r$.
\State Construct $P_{\textup{cand}} = P_{\textup{lattice}} \cap [\rho_{\myvar{e}_1}^{\min}-r/2,\rho_{\myvar{e}_1}^{\max}+r/2]\times [\rho_{\myvar{e}_2}^{\min}-r/2,\rho_{\myvar{e}_2}^{\max}+r/2] \times ... \times [\rho_{\myvar{e}_n}^{\min}-r/2,\rho_{\myvar{e}_n}^{\max}+r/2]  $.
\State $P = \{ \myvar{p} \in P_{\textup{cand}}:  B(\myvar{p},r)\cap \myset{S}\neq \emptyset \}, G = P\times\{r\}$.
\State \textbf{return} $G$. 

\end{algorithmic}
\end{algorithm}

\begin{proposition} \label{prop:grid sampling}
Given a compact set $\myset{S}\subset\mathbb{R}^n$and a lattice size $r>0$, then the following hold:
\begin{itemize}
    \item[1)] $G$, from Algorithm \ref{alg:grid sampling}, is of finite cardinality, and
    \item[2)]  $\myset{S}\subseteq \cup_{\myvar{p}\in P} B(\myvar{p},r) $, where $P$ is given in Algorithm \ref{alg:grid sampling}, Line 4.
\end{itemize}
\end{proposition}
\begin{proof}
    Since the set $\myset{S}$ is compact, the lower and upper range limit $\rho_{\myvar{e}_i}^{\min}$ and  $\rho_{\myvar{e}_i}^{\max}$, $i = 1,2, .., n$, given in Line 1 of Algorithm \ref{alg:grid sampling}, are finite for every dimension. This leads to the fact that    the hyperrectangle $[\rho_{\myvar{e}_1}^{\min}-r/2,\rho_{\myvar{e}_1}^{\max}+r/2]\times [\rho_{\myvar{e}_2}^{\min}-r/2,\rho_{\myvar{e}_2}^{\max}+r/2] \times ... \times [\rho_{\myvar{e}_n}^{\min}-r/2,\rho_{\myvar{e}_n}^{\max}+r/2] $ is bounded. Recall that by definition \eqref{eq:lattice}, $P_{\textup{lattice}}$ denotes a regular lattice in $\mathbb{R}^n$, and we thus know that
    $P_{\textup{cand}}$ has a finite cardinality, which implies that $G$  also has a finite cardinality. Now we show Property 2) by contradiction. Assume that there exists $\myvar{x}\in \myset{S}$ and $\myvar{x}\notin  \cup_{\myvar{p}\in P} B(\myvar{p},r) $. In view of the definition of $P$, this implies that $\myvar{x}\notin  \cup_{\myvar{p}\in P_{\textup{cand}}} B(\myvar{p},r)$.  This yields a contradiction since $\myvar{x}\in \myset{S}\subseteq [\rho_{\myvar{e}_1}^{\min},\rho_{\myvar{e}_1}^{\max}]\times [\rho_{\myvar{e}_2}^{\min},\rho_{\myvar{e}_2}^{\max}] \times ... \times [\rho_{\myvar{e}_n}^{\min},\rho_{\myvar{e}_n}^{\max}] \subseteq    \cup_{\myvar{p}\in P_{\textup{cand}}} B(\myvar{p},r)$.  The former set inclusion is trivial in view of the definition of $\rho_{\myvar{e}_i}^{\min},\rho_{\myvar{e}_i}^{\max}$. The latter set inclusions can be straightforwardly checked by discussing all possible relations of the points in $P_{cand}$ and the hyberrectangle. 
\end{proof}

From now on, we denote $\textup{Bound}(\myset{S})$ the bounding box $[\rho_{\myvar{e}_1}^{\min}-r/2,\rho_{\myvar{e}_1}^{\max}+r/2]\times [\rho_{\myvar{e}_2}^{\min}-r/2,\rho_{\myvar{e}_2}^{\max}+r/2] \times ... \times [\rho_{\myvar{e}_n}^{\min}-r/2,\rho_{\myvar{e}_n}^{\max}+r/2] $ of a compact set $\myset{S}$.

\begin{example}

\begin{figure}
    \centering
    \includegraphics[width=\linewidth]{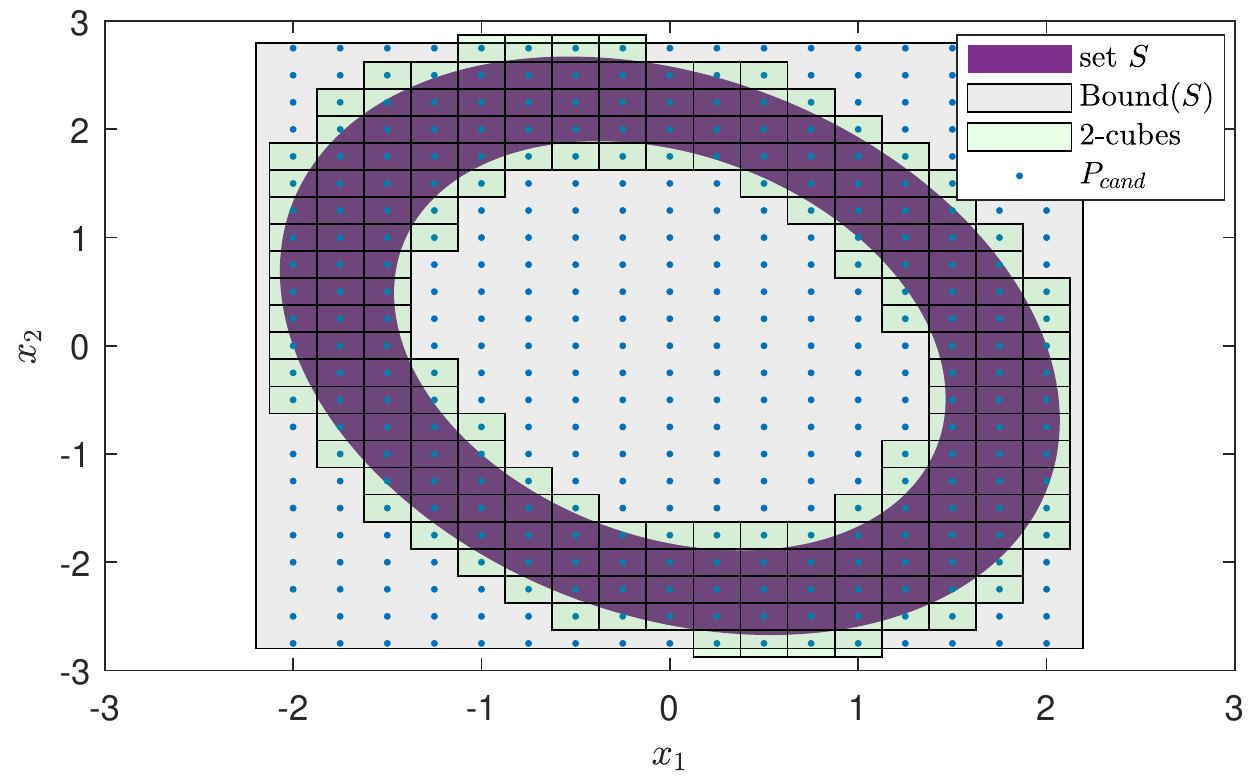}
    \caption{Grid sampling of a set $\myset{S}$ in 2-D. Here the set $\myset{S}$ is shown in violet, $\textup{Bound}(\myset{S})$ is shown in gray, the $2$-cubes generated from Algorithm \ref{alg:grid sampling} are in light green, and all the points in blue form $P_{\textup{cand}}$. We observe that $\myset{S}\subset \textup{Bound}(\myset{S}) $ and $\myset{S}$ is over-approximated by the union of the $2$-cubes. }
    \label{fig:gridsampling}
\end{figure}
Here we show an example of Algorithm \ref{alg:grid sampling} with the set $\myset{S} = \{ \myvar{x}\in \mathbb{R}^2: 1\leq \myvar{x}^\top Q\myvar{x} \leq 2, \textup{ where } Q = \begin{psmallmatrix}
0.5 & 0.1 \\
0.1 & 0.3\end{psmallmatrix}\}$ and $r = 0.25$. From Fig. \ref{fig:gridsampling}, we observe that $G$ has a finite cardinality and $\myset{S}\subseteq \cup_{\myvar{p}\in P} B(\myvar{p},r) $. It is worth noting that $\cup_{\myvar{p}\in P} B(\myvar{p},r) \nsubseteq \textup{Bound}(\myset{S}) $, where $P$ is given in Algorithm \ref{alg:grid sampling} Line 4, as shown in Fig. \ref{fig:gridsampling}. We also note that the intersection condition in Line $4$ can be checked numerically as a feasibility problem. 
\end{example}

\subsection{Proposed verification algorithm}

 Now consider the compatibility verification problem in \eqref{eq:feasibility problem}. For any $\myvar{x}\in \textup{Bound}(\myset{\myset{C}})$, define
 \begin{equation} \label{eq:c_x}
     \begin{aligned}
         & c(\myvar{x}) = \max_{\myvar{u},t} t \\
         & \textup{s.t. } A(\myvar{x}) \myvar{u} + b(\myvar{x}) \geq t \myvar{1}_{N}, \\
         & \hspace{20pt} \myvar{u} \in \mathbb{U}.
     \end{aligned}
 \end{equation}
 In the case that $\mathbb{U}$ is a polytopic set,  $c(\myvar{x})$ is obtained by solving a linear program. In the general case where $\mathbb{U}$ is convex, $c(\myvar{x})$ is obtained from a convex optimization. One interpretation is that $c(\myvar{x})$ indicates the largest robustness level at $\myvar{x}$ up to which the CBF conditions or the input constraints are to be breached.
 
 Recall that the candidate CBFs $h_i(\myvar{x}), i = 1,2,..., N$ are continuously differentiable, the vector fields $\myvarfrak{f}(\myvar{x})$ and $\myvarfrak{g}(\myvar{x})$ are locally Lipschitz, and thus
$A(\myvar{x}), b(\myvar{x})$ in \eqref{eq:feasibility problem} are locally Lipschitz.  Specifically, denote the respective Lipschitz constants in the bounding box $\textup{Bound}(\myset{C})$ with respect to the $l_{\infty}$ norm as $L_{A,\infty}, L_{b,\infty}$, i.e.,
\begin{equation} \label{eq:lipschitz constant}
\begin{aligned}
   & \| A(\myvar{x}) - A(\myvar{x}^\prime)\|_{\infty} \leq L_{A,\infty} \| \myvar{x} - \myvar{x}^\prime \|_{\infty}, \\
  & \| b(\myvar{x}) - b(\myvar{x}^\prime)\|_{\infty} \leq L_{b,\infty} \| \myvar{x} - \myvar{x}^\prime \|_{\infty},
\end{aligned}
\end{equation}
for all $\myvar{x},\myvar{x}^\prime \in \textup{Bound}(\myset{C})$\footnote{Here $\| A\|_{\infty}$, where $A$ is a matrix, refers to the induced matrix norm and can be calculated as the maximum absolute row sum of $A$.}. 

If $c(\myvar{x})>0$ for some $\myvar{x}$, based on the Lipschitz continuity of $A(\myvar{x})$ and $b(\myvar{x})$, there must exist a neighborhood around $\myvar{x}$ where the CBFs $h_i(\myvar{x})$ are compatible. This  is formally shown below.

\begin{proposition} \label{prop:probe and expand}
For any $\myvar{x}\in \textup{Bound}(\myset{\myset{C}})$, if $c(\myvar{x})>0$, then $
\sup_{\myvar{v}\in \mathbb{U}}A(\myvar{x}^\prime)\myvar{v}+b(\myvar{x}^\prime) \ge \myvar{0}$ for all $\myvar{x}^\prime \in    B(\myvar{x},\rho(\myvar{x}))\cap \textup{Bound}(\myset{\myset{C}})$ with \begin{equation} \label{eq:rho}
    \rho(\myvar{x})  = \frac{ 2c(\myvar{x})}{L_{A,\infty}\| \myvar{u}^\star (\myvar{x}) \|_{\infty} + L_{b,\infty}},
\end{equation}
where $L_{A,\infty}, L_{b,\infty}$ are the Lipschitz constants of $A(\myvar{x}), b(\myvar{x})$ with respect to the $l_{\infty}$ norm as per \eqref{eq:lipschitz constant}, respectively, and $\myvar{u}^\star(\myvar{x})$ is the optimal solution to \eqref{eq:c_x} at $\myvar{x}$.
\end{proposition} 

\begin{proof}
    For any $\myvar{x}^\prime \in B(\myvar{x},\rho)\cap \textup{Bound}(\myset{\myset{C}})$, we have
    \begin{multline}
 \label{eq:Ax'+bx'}
        A(\myvar{x}^\prime)\myvar{u}^\star(\myvar{x}) + b(\myvar{x}^\prime ) 
        = (A(\myvar{x}^\prime)  - A(\myvar{x})) \myvar{u}^\star(\myvar{x}) \\ + (b(\myvar{x}^\prime) - b(\myvar{x})) + A(\myvar{x})\myvar{u}^\star(\myvar{x})+ b(\myvar{x})
    \end{multline}
    In view of \eqref{eq:lipschitz constant}, we have 
    \begin{equation} \label{eq:worst case error}
    \begin{aligned}
       & \|(A(\myvar{x}^\prime) - A(\myvar{x})) \myvar{u}^\star(\myvar{x}) + b(\myvar{x}^\prime) - b(\myvar{x}) \|_{\infty}\\
        & \leq \|A(\myvar{x}^\prime) - A(\myvar{x}) \|_{\infty} \|\myvar{u}^\star(\myvar{x}) \|_{\infty} + \|b(\myvar{x}^\prime) - b(\myvar{x}) \|_{\infty} \\
        & \leq L_{A,\infty} \| \myvar{x} - \myvar{x}^\prime \|_{\infty}  \|\myvar{u}^\star(\myvar{x}) \|_{\infty} +  L_{b,\infty} \| \myvar{x} - \myvar{x}^\prime\|_{\infty}
    \end{aligned}
    \end{equation}
    In view of $\myvar{x}^\prime \in B(\myvar{x},\rho)$, and $\rho$ in \eqref{eq:rho},  we obtain $\| \myvar{x} - \myvar{x}^\prime\|_{\infty}\leq \rho/2 = \frac{ c(\myvar{x})}{L_{A,\infty}\| \myvar{u}^\star(\myvar{x}) \|_{\infty} + L_{b,\infty}}$. Thus, $ \|(A(\myvar{x}^\prime) - A(\myvar{x})) \myvar{u}^\star(\myvar{x}) + b(\myvar{x}^\prime) - b(\myvar{x}) \|_{\infty} \leq c(\myvar{x})$. From \eqref{eq:Ax'+bx'} and $A(\myvar{x})\myvar{u}^\star(\myvar{x})+ b(\myvar{x})\ge c(\myvar{x})\myvar{1}$, we further obtain $A(\myvar{x}^\prime)\myvar{u}^\star + b(\myvar{x}^\prime )\ge \myvar{0}$, which completes the proof.
\end{proof}

\begin{algorithm}[h]
\caption{\texttt{CompatibilityChecking}} \label{alg:compatibility checking}
\begin{algorithmic}[1]
\Require $h_i(\myvar{x}),\alpha_i(\cdot)$, initial size $r_0$, decaying factor $\lambda$
\State \textbf{Initialization:}
\State  \hspace{12pt} $k = 0$, obtain $\myset{C} $ from \eqref{eq:set_c}, $G_{0} \leftarrow \texttt{GS}(\myset{C}, r_0)$, $ G_{1} = \emptyset $.
\While{$G_{k}\neq \emptyset$}
\For{each $ (\myvar{x},r)\in  G_{k} $}
\State  $c \leftarrow c(\myvar{x})$ from \eqref{eq:c_x}, $ \rho \leftarrow \rho(\myvar{x})$ from \eqref{eq:rho}.
\If{$c<0$}  \Comment{Found an incompatible state;}
    \State \textbf{return False}. 
\ElsIf{$\rho\ge r$} \Comment{Compatibility checked;}
    \State remove $(\myvar{x},r)$ from $G_{k} $. 
\Else \Comment{Compatibility partially checked;}
\State remove $(\myvar{x},r)$ from $G_{k} $, $r^\prime\leftarrow \lambda r$.
\State $G_{k+1} \leftarrow G_{k+1} \cup  \texttt{GS}(B(\myvar{x},r)\setminus B(\myvar{x},\rho), r^\prime)$.
\EndIf
\EndFor
\State $ k = k+1, G_{k+2} = \emptyset$.
\EndWhile
\State \textbf{return True}.
\end{algorithmic}
{ \small *\texttt{GS} stands for \texttt{GridSampling} given in Algorithm \ref{alg:grid sampling}.}
\end{algorithm}

Built on above analysis, we design a  compatibility checking algorithm  using grid sampling and refinement. As given in Algorithm \ref{alg:compatibility checking}, the safety set $\myset{C}$ is firstly over-approximated using \texttt{GridSampling} Algorithm with an initial lattice size $r_0$. This will yield a finite set $G_0$ of $n$-cubes that is to be checked later. Recall that in Problem formulation \eqref{eq:feasibility problem} and \eqref{eq:robust feasibility problem}, we need to check the compatibility over a set $\myset{D}\supseteq \myset{C}$. Here we take $\myset{D} = \cup_{(\myvar{x}_i,r_0) \in G_0} B(\myvar{x}_i,r_0) $, which is a super set of $\myset{C}$ from Proposition \ref{prop:grid sampling}, item 2).  Choosing $r_0$ is important and depends on how large buffering zone one allows outside the safety set. For each $n$-cube $B(\myvar{x},r)$ in $G_k$, represented as a $(\myvar{x},r)$ pair in Line 4, we calculate the robustness level $c$ and the size $\rho$ of a guaranteed compatible $n$-cube centered at $\myvar{x}$ from \eqref{eq:c_x} and \eqref{eq:rho}, respectively. If $c<0$, then an incompatible state is found and the algorithm terminates and returns \texttt{False}. If $\rho \ge r$, then we know that the CBFs are compatible for all the states within the $n$-cube $B(\myvar{x},r)$ and we remove $(\myvar{x},r)$ from $G_k$; otherwise, we refine the remaining unchecked region $B(\myvar{x},r)\setminus B(\myvar{x},\rho)$ with a discounted lattice size $r^\prime = \lambda r$ and include the new $n$-cubes in $G_{k+1}$. After checking all the $n$-cubes in $G_k$, we iterate the process again for $G_{k+1}$. Once $G_{k+1}=\emptyset$, the algorithm terminates and returns \texttt{True}.

The following properties provide a guarantee on the finite-step termination of the algorithm and the compatibility property certified from its termination.
\begin{thm} \label{thm:compatibility}
Given control barrier functions $h_i(\myvar{x})$, extended class $\mathcal{K}$ functions $\alpha_i(\cdot)$  with $i \in \mathcal{I}$, an initial lattice size $r_0 >0$ and a decaying factor $0<\lambda<1$, we have:
\begin{enumerate}
    \item If Algorithm \ref{alg:compatibility checking} terminates, it gives verification or falsification on the CBF compatibility as per Def. \ref{def:compatible_cbf};
    \item if $\mathbb{U}$ is bounded, and the CBFs $h_i(\myvar{x})$ are robustly compatible with robustness level $\eta>0$ in $\textup{Bound}(\myset{C})$, then  Algorithm \eqref{alg:compatibility checking} terminates in finite steps.
    \item If $\mathbb{U}$ is bounded, and a lower bound of the lattice size $\underline{r}$ is incorporated, i.e., Algorithm \ref{alg:compatibility checking} terminates if $r\leq \underline{r}$ in Line 4, then Algorithm \ref{alg:compatibility checking} terminates in finite steps and gives  one of the following three results:
    \begin{itemize}
        \item[i.] $h_i(\myvar{x}), i\in \mathcal{I}$ are compatible;
        \item[ii.] $h_i(\myvar{x}), i\in \mathcal{I}$ are  incompatible; 
        \item[iii.] $h_i(\myvar{x}), i\in \mathcal{I}$ are not robustly compatible with robust level greater than
        \begin{equation} \label{eq:eta_prime}
            \eta^\prime= \lambda^{-1} \underline{r}( \max_{ \myvar{u}\in \mathbb{U}} L_{A,\infty}\| \myvar{u} \|_{\infty} + L_{b,\infty})/2.
        \end{equation}
    \end{itemize}
\end{enumerate}
\end{thm}

\begin{proof}
From Proposition \ref{prop:grid sampling}, we know that the $n$-cubes generated by the grid sampling algorithm will over-approximate the safety set $\myset{C}$ (Line 2) or the remaining unchecked region $B(\myvar{x},r)\setminus B(\myvar{x},\rho)$ (Line 12) with a finite number of $n$-cubes.  If Algorithm \ref{alg:compatibility checking} terminates, it indicates that either an incompatible state is found or a  state set over-approximating the safety set has been checked for compatibility. This proves Property 1).

Now consider the case when the CBF $h_i(\myvar{x})$s are robustly compatible with level $\eta$, i.e., $c(\myvar{x})\ge \eta, \forall \myvar{x}\in \textup{Bound}(\myset{C})$. At the $k$th iteration, $r_k = \lambda^k r_0$ is the size of the $n$-cubes in $G_k$. Denote  $ \underline{\rho} = \min_{ \myvar{u}\in \mathbb{U}} \frac{ 2\eta}{L_{A,\infty}\| \myvar{u} \|_{\infty} + L_{b,\infty}}$. As $\mathbb{U}$ is assumed to be bounded, $ \underline{\rho} >0$. As $r_k$ is decreasing exponentially fast to zero as $k$ grows, we deduce there exists a $M\in \mathbb{N}$ such that $r_M \leq \underline{\rho}$.  Thus, Algorithm \ref{alg:compatibility checking} will terminate in finite steps. This completes the proof of Property 2).

If a lower bound of the lattice size $\underline{r}$ is incorporated, Algorithm \ref{alg:compatibility checking}  terminates in one of the following three cases: \textbf{Case a}. it has found an incompatible state and returns \texttt{False} (Line 7); \textbf{Case b}. it has checked all the $n$-cubes and returns \texttt{True} (Line 17); \textbf{Case c}. $r_M\le \underline{r}, M\in \mathbb{N}$ for the first time as the algorithm iterates, i.e., $r_0> \underline{r}, r_1> \underline{r}, ..., r_{M-1}> \underline{r}, r_M\le \underline{r}$, where $r_k, k\in \{0,1,2,...\}$ is the size of $n$-cubes in $G_k$ at the $k$th iteration (Line 4). These three termination cases correspond to the three possible results in Property 3). Since $r_k = \lambda^k r_0$, we deduce that such a $M\in \mathbb{N}$ exists. Thus, the finite-step termination is concluded since the algorithm executes at most $M $ iterations. 

Now we show Property 3) Result iii by contradiction. Assume that $h_i(\myvar{x}), i\in \mathcal{I}$ are robustly compatible with level $\eta^\prime$ in \eqref{eq:eta_prime}, then based on the analysis for Property 2), we know there exists a constant lattice size $\underline{\rho}^\prime  = \min_{ \myvar{u}\in \mathbb{U}} \frac{ 2\eta^{\prime}}{L_{A,\infty}\| \myvar{u} \|_{\infty} + L_{b,\infty}}$ and  if the size of the remaining $n$-cubes is smaller than $\underline{\rho}^\prime$, they can be checked in one iteration (Line 8). Substituting $\eta^\prime$ in \eqref{eq:eta_prime}, we have $\underline{\rho}^\prime = \lambda^{-1} \underline{r}$. Since  $ r_{M}\leq \underline{r}$, we have $r_{M-1} = \lambda^{-1}r_{M} \leq \underline{\rho}^\prime $ due to $\underline{\rho}^\prime = \lambda^{-1} \underline{r}$. This, however, contradicts with the termination condition of \textbf{Case c} that $r_{M-1}> \underline{r}$. Thus, the termination \textbf{Case c} occurs only if  $h_i(\myvar{x})$s are not robustly compatible with robustness level greater than $\eta^\prime$. This concludes the proof.
\end{proof}

\section{Discussions}

\subsection{Computational concerns}

As any lattice-based verification method, one main computational issue of Algorithm \ref{alg:compatibility checking} is the exponential growth of the number of the $n$-cubes as the system dimension grows and thus this approach is limited to low-dimensional systems.  We propose the following to mitigate the computational burdens.

Recall that to prove the forward invariance of the safety set, we only need to guarantee the CBF conditions around the safety boundary. This relaxation of the classic CBF conditions in \eqref{eq:cbf} has been explored in \cite{xiao2021high,cortez2021robust}. Thus, one way to reduce the number of $n$-cubes is to check the CBF compatibility only around the safety boundary. Specifically, we can change $G_{0} \leftarrow \texttt{GS}(\myset{C}, r_0)$ in Algorithm \ref{alg:compatibility checking} Line 2 to $G_{0} \leftarrow \texttt{GS}(\myset{C} \setminus \myset{C}_a, r_0)$, where $\myset{C}_a = \{\myvar{x}\in \mathbb{R}^n: h_i(\myvar{x})\ge a, a>0, \forall i \in \mathcal{I}\}$. By this modification, we do not need to check a region that lies strictly in the interior of the safety set $\myset{C}$. Another interpretation of this modification is that we can always choose \text{appropriate} $\alpha_i(\cdot)$ so that the third term in the CBF condition \eqref{eq:cbf} dominates the first two terms for all $\myvar{x}\in \myset{C}_a$ when the input set $\mathbb{U}$ is bounded and $\myset{C}$ is compact.

We also note that \texttt{CompatibilityChecking} is computed offline, irrelevant to the nominal control design, and, for each iteration, the process for each $n-$cube (Line 4 to Line 14 in Algorithm \ref{alg:compatibility checking}) can be executed in parallel. 

\subsection{Generalization to time-varying dynamics and safety set}
One basic setup in previous sections is that the control system \eqref{eq:nonlinear_dyn} and the safety set \eqref{eq:set_c} are time-invariant. The result can be trivially generalized to time-varying dynamics $\dot{\myvar{x}} = \mathfrak{f}(\myvar{x},t) + \mathfrak{g}(\myvar{x},t) \myvar{u}$ with a time-varying safety set $\myset{C}(t) = \{\myvar{x}\in \mathbb{R}^n: h_i(\myvar{x},t) \ge 0, \forall i \in  \mathcal{I}\}$ as follows. Let $\tilde{x} : = (\myvar{x},t)$ be a new state variable. Thus, we obtain the new dynamics $\dot{\tilde{\myvar{x}}} = \begin{psmallmatrix}
\mathfrak{f}(\tilde{\myvar{x}})\\
1
\end{psmallmatrix} + \begin{psmallmatrix}
\mathfrak{g}(\tilde{\myvar{x}}) \\
\myvar{0}
\end{psmallmatrix}\myvar{u}$ and the new safety set $\tilde{\myset{C}} = \{\tilde{\myvar{x}}\in \mathbb{R}^{n+1}: h_i(\tilde{\myvar{x}}) \ge 0, \forall i \in \mathcal{I} \}$. Due to Assumption \ref{ass:compact C}, however, we can only consider a bounded time interval.

\subsection{Alternative grid sampling methods}
In Algorithm \ref{alg:compatibility checking}, we have utilized the grid sampling algorithm (Algorithm \ref{alg:grid sampling}) which, for any compact set $\myset{S}\subset \mathbb{R}^n$, generates a finite number of $n$-cubes whose union over-approximates set $\myset{S}$. There are of course alternative grid sampling methods. One option is to use $n$-spheres $ B_S(\myvar{x},r) = \{ \myvar{y}\in \mathbb{R}^{n}: \| \myvar{y} - \myvar{x} \|\leq r \}$. Furthermore, we can show that,  in a similar manner to the proof of Proposition \ref{prop:probe and expand}, if $c(\myvar{x})>0$, then there exists a radius $\rho_S(\myvar{x})$ depending on $c(\myvar{x})$ and the Lipschitz constants $L_A, L_b$ of $A(\myvar{x}), b(\myvar{x})$ with respect to the $l_2$ norm such that the CBF conditions are compatible  for all $ B_S(\myvar{x},\rho_S)$. Thus, a similar algotithm as Algorithm \ref{alg:compatibility checking} and a similar theoretical result as Theorem \ref{thm:compatibility} can be developed.

Despite that choosing $n$-spheres for grid sampling is viable, we opt for $n$-cubes for the following reasons: 1) for an arbitrarily compact set $\myset{S}\subset \mathbb{R}^n, n>3$, it is generally difficult to generate a set of $n$-spheres whose union covers $\myset{S}$ with a small overlapping ratio; 2) since each row in $A(\myvar{x})$ and $b(\myvar{x})$ correspond to one safety constraint, we can calculate $L_{A,\infty}, L_{B,\infty}$ by checking each of the constraints individually, rendering the computation easier in general. A more detailed calculation of these Lipschitz constants is given in Example \ref{ex:example2} below.

\subsection{Other improvements}
There are also other heuristics to improve Algorithm \ref{alg:compatibility checking}. For example, instead of using the Lipschitz constants $L_{A,\infty}, L_{b,\infty}$ in $\textup{Bound}(\myset{C})$ in Line 5 of Algorithm \ref{alg:compatibility checking}, we can calculate a more precise $L_{A,\infty}, L_{b,\infty}$ in $B(\myvar{x},r)$ and then calculate $\rho(\myvar{x})$. This would lead to a larger $\rho$ and fewer iterations in general.

Another possible improvement is about the updated lattice size. It is reasonable to assume that $\rho(\myvar{x})$ will not vary too much in a neighborhood of $\myvar{x}$, thus the updated size $r^\prime$ in Line 11 of Algorithm \ref{alg:compatibility checking} could be upper bounded by $\rho(\myvar{x})$, i.e., $r^\prime \leftarrow \min(\rho,\lambda r)$ in place of $r^\prime \leftarrow \lambda r$. We note that this does not affect the theoretical results in Theorem \ref{thm:compatibility}.

\section{Case studies}
In this section we show more details on the algorithm implementation, especially the Lipschitz constant calculation, and demonstrate the efficacy of our proposed verification algorithm in several different scenarios. All the simulations are done using Matlab Parallel Computing Toolbox on an Intel i7-8650U CPU laptop.

\begin{figure*}[h!]
	\centering
	\begin{subfigure}[t]{0.3\linewidth}
		\includegraphics[width=\linewidth]{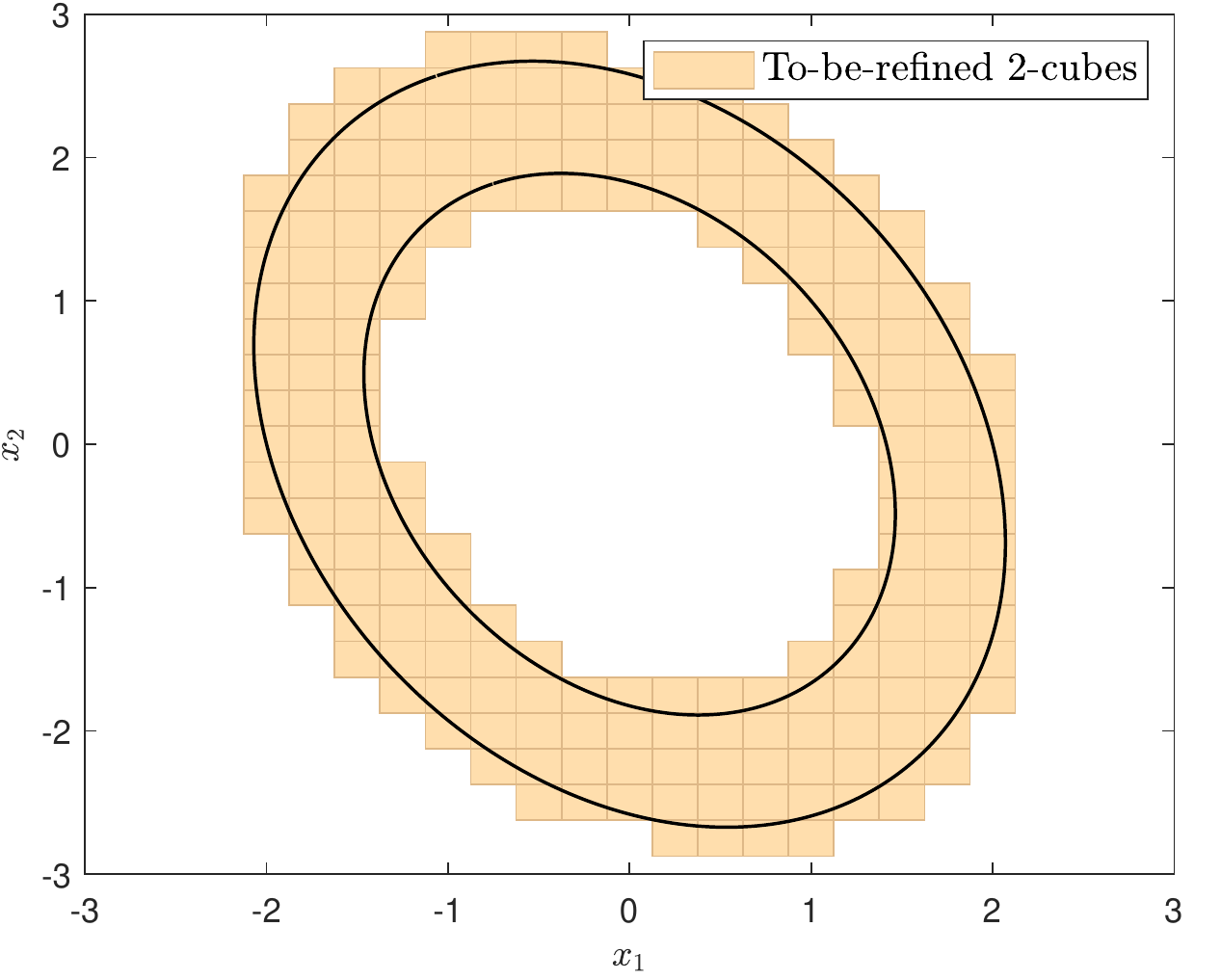}
		\caption{   First iteration, $r = 0.25$. }   
	\end{subfigure} %\hspace{5mm}
	\begin{subfigure}[t]{0.3\linewidth}
		\centering\includegraphics[width=\linewidth]{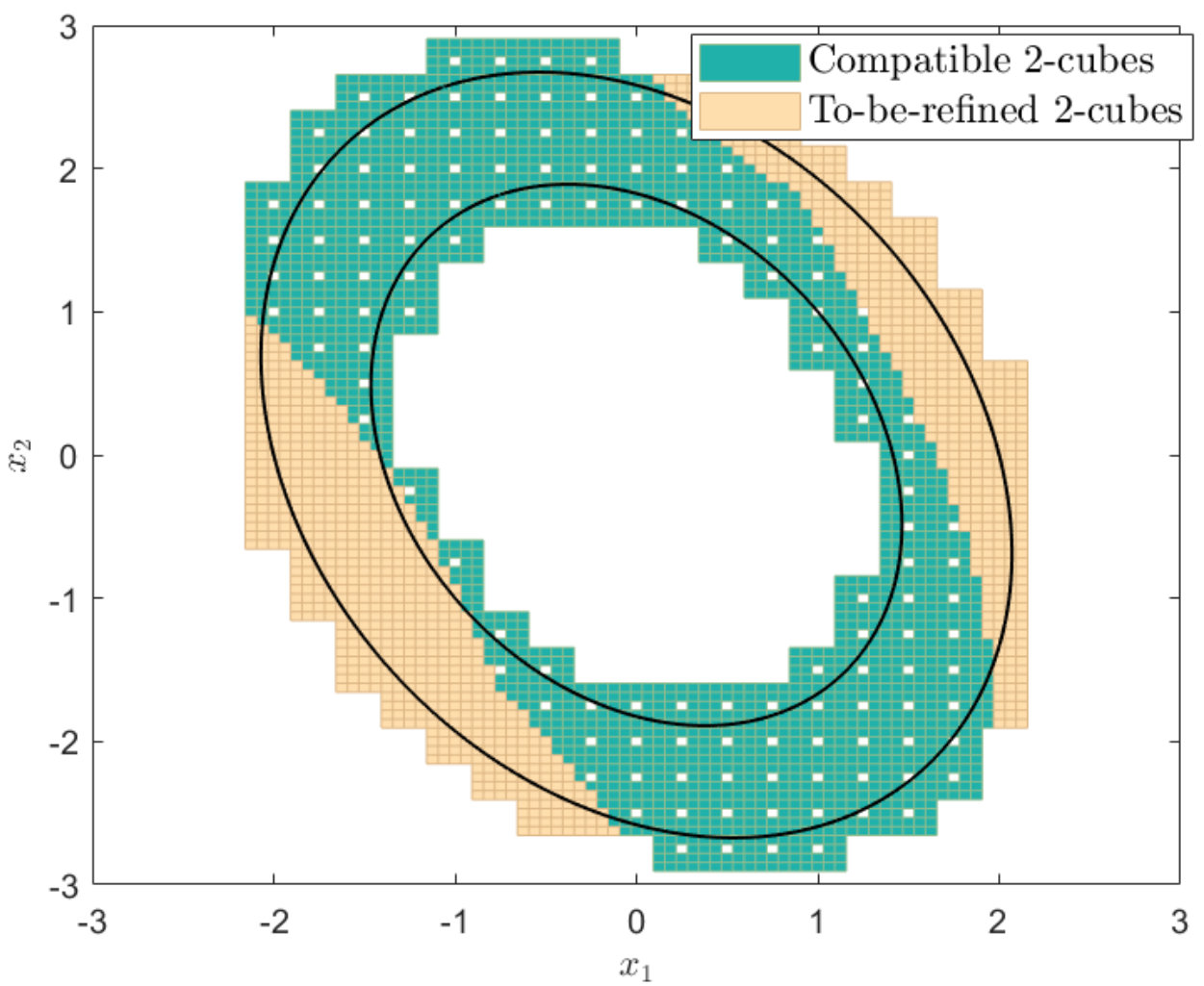}
		\caption{  Second iteration, $r = 0.0625$.}
	\end{subfigure} 
% 	\hspace{5mm}
	\begin{subfigure}[t]{0.3\linewidth}
		\centering\includegraphics[width=\linewidth]{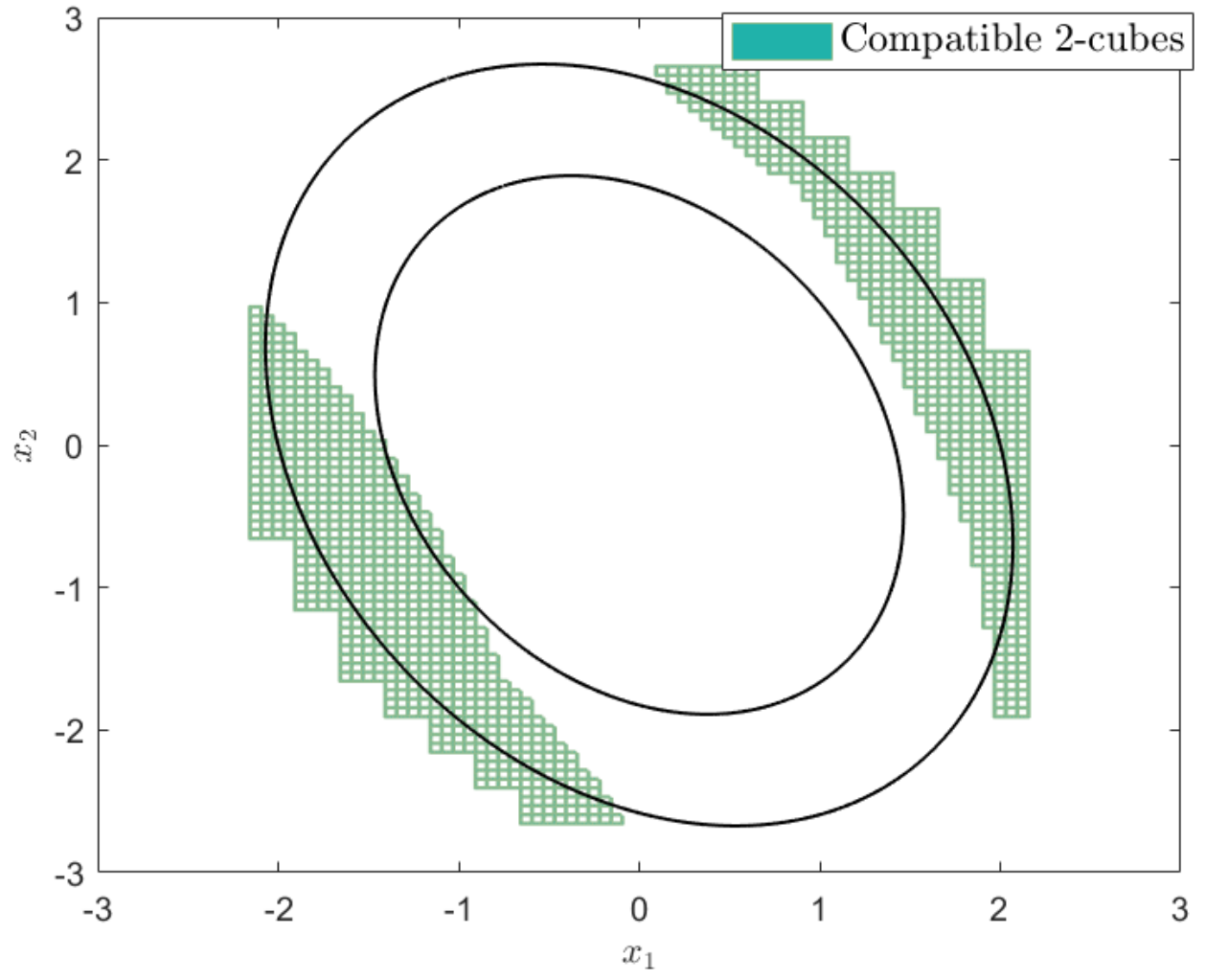}
		\caption{  Third iteration,  $r=0.0156$.}
	\end{subfigure}
	\caption{ Execution process of \texttt{CompatibilityChecking} in Example 2. The safety region is between the two ellipsoids. Compatible $2$-cubes: $2$-cubes within which the CBFs are verified to be compatible (in green); to-be-refined $2$-cubes: $2$-cubes that need further refinement (in yellow). (a) First iteration with the lattice size $r = 0.25$. All of the total $200$ $2$-cubes are to-be-refined $2$-cubes. (b) Second iteration with the lattice size $r = 0.0625$. The refined $2$-cubes are checked and $3204$ out of $4869$ are verified to be compatible $2$-cubes, and $1665$ $2$-cubes are refined again. (c) Third iteration with the lattice size $r = 0.0156$. All the $2$-cubes are compatible. Algorithm \ref{alg:compatibility checking} thus gives verification on the compatibility of the CBFs.  }  
	\label{fig:example2}	
	\end{figure*}

\begin{example}\label{ex:example2}
Consider a $2-D$ system with state variable $\myvar{x} = (x_1,x_2)$, input variable $ \myvar{u} = (u_1,u_2)$, dynamics
\begin{equation}
    \begin{pmatrix}
    \dot{x}_1 \\
    \dot{x}_2
    \end{pmatrix} = 
    \underbrace{\begin{pmatrix}
    x_1+x_2 \\
    -x_1^2/2
    \end{pmatrix}}_{\mathfrak{f}(\myvar{x})} + \underbrace{\begin{pmatrix}
    1 & 0 \\
    0 & 1
    \end{pmatrix}}_{\mathfrak{g}(\myvar{x})}\begin{pmatrix}
    u_1 \\
    u_2
    \end{pmatrix},
\end{equation}
and the input constraint set $\mathbb{U} = \{ (u_1,u_2): |u_1|\leq 3, |u_2| \leq 3\}.$ The two CBF candidates are $h_1(\myvar{x}) =  \myvar{x}^\top Q\myvar{x} -1,
    h_2(\myvar{x}) =  2- \myvar{x}^\top Q\myvar{x},$
% \begin{align}
%     h_1(\myvar{x}) & =  \myvar{x}^\top Q\myvar{x} -1\\
%     h_2(\myvar{x}) & =  2- \myvar{x}^\top Q\myvar{x}
% \end{align}
where $Q = \begin{psmallmatrix}
0.5 & 0.1 \\
0.1 & 0.3\end{psmallmatrix}$. The safety set is $\myset{C} = \{ \myvar{x}: h_i(\myvar{x})\ge 0, i = 1,2\}$. The corresponding extended class $\mathcal{K}$ functions are chosen as $\alpha_1(v) = v, \alpha_2(v) = v, v\in \mathbb{R}$. The CBF conditions are then given by
\begin{equation*}
    A(\myvar{x}):=\begin{pmatrix}
\nabla^\top h_1(\myvar{x})  \\
\nabla^\top h_2(\myvar{x})
\end{pmatrix}, b(\myvar{x}):=\begin{pmatrix}
\nabla^\top h_1(\myvar{x})\mathfrak{f}(\myvar{x})+ h_1(\myvar{x}) \\
\nabla^\top h_2(\myvar{x})\mathfrak{f}(\myvar{x})+ h_2(\myvar{x}) 
\end{pmatrix}. 
\end{equation*}
where $\nabla h_1(\myvar{x})=(Q+Q^\top)\myvar{x} = 2Q\myvar{x}, \nabla h_2(\myvar{x})=-2Q\myvar{x} $. Choose the initial lattice size $r_0 = 0.25$. By applying \texttt{GridSampling}$(C,r_0)$, we obtain,  as shown in Fig. \ref{fig:gridsampling}, $\textup{Bound}(\myset{C}) = [-2.2,2.2]\times[-2.8,2.8]$.

Now we calculate the Lipschitz constants $L_{A,\infty}, L_{b,\infty}$ in $\textup{Bound}(\myset{C})$. We note that the Lipschitz constants can be obtained by considering each CBF individually, taking advantage of the fact that each row of $A(\myvar{x})$ and $b(\myvar{x})$ corresponds to one CBF. By definition,  $ L_{A,\infty}$ needs to satisfy
\begin{equation} \label{eq:L_Ainf condition}
    \| A(\myvar{x}) - A(\myvar{x}^\prime)\|_{\infty} = \| \begin{psmallmatrix}
    2(\myvar{x} - \myvar{x}^\prime)^\top Q \\
    - 2(\myvar{x} - \myvar{x}^\prime)^\top Q 
    \end{psmallmatrix}\|_{\infty} \leq L_{A,\infty}\| \myvar{x} - \myvar{x}^\prime\|_{\infty}
\end{equation}
for any $\myvar{x},\myvar{x}^\prime $ in $\textup{Bound}(\myset{C})$. Let $\myvar{a}_i(\myvar{x})$ be the $i$th row of $A(\myvar{x})$. Condition \eqref{eq:L_Ainf condition} is equivalent to
\begin{multline}
    \| \myvar{a}_i(\myvar{x}) - \myvar{a}_i(\myvar{x}^\prime)\|_1 = \| 2Q(\myvar{x} - \myvar{x}^\prime)\|_1  \\ \leq L_{A,\infty}\| \myvar{x} - \myvar{x}^\prime\|_{\infty}, \forall i\in \{1,2\}.
\end{multline}
This is implied by the condition $   \max_{\myvar{x},\myvar{x}^\prime} \frac{ \| 2Q(\myvar{x} - \myvar{x}^\prime)\|_1}{\| \myvar{x} - \myvar{x}^\prime\|_{\infty}} \leq L_{A,\infty}$.
% \begin{equation}
%   \max_{\myvar{x},\myvar{x}^\prime} \frac{ \| 2Q(\myvar{x} - \myvar{x}^\prime)\|_1}{\| \myvar{x} - \myvar{x}^\prime\|_{\infty}} \leq L_{A,\infty}
% \end{equation}
Using the inequality $   \| \myvar{v}\|_{\infty} \leq \| \myvar{v}\|_1 \leq n\|\myvar{v}\|_{\infty}, \forall \myvar{v}\in \mathbb{R}^n,$
% \begin{equation}
%     \| \myvar{v}\|_{\infty} \leq \| \myvar{v}\|_1 \leq n\|\myvar{v}\|_{\infty}, \forall \myvar{v}\in \mathbb{R}^n,
% \end{equation}
 we have
 \begin{multline}
     \max_{\myvar{x},\myvar{x}^\prime, \myvar{x}\neq \myvar{x}^\prime} \frac{ \| 2Q(\myvar{x} - \myvar{x}^\prime)\|_1}{\| \myvar{x} - \myvar{x}^\prime\|_{\infty}}  \\ \leq 2 \max_{\myvar{x},\myvar{x}^\prime,\myvar{x}\neq \myvar{x}^\prime} \frac{ \| 2Q(\myvar{x} - \myvar{x}^\prime)\|_1}{\| \myvar{x} - \myvar{x}^\prime\|_{1}} = 4\|Q\|_1 
 \end{multline}
The equality holds due to the definition of induced matrix norm. This reveals that $L_{A,\infty} = 4\| Q\|_1 = 2.4 $ satisfies \eqref{eq:L_Ainf condition}.\footnote{Here $\| Q\|_{1}$, where $Q$ is a matrix, refers to the induced matrix norm and can be calculated as the maximum absolute column sum of $Q$.}
 
 Similarly, $ L_{b,\infty}$ needs to satisfy
\begin{equation} \label{eq:L_binf condition}
    \| b(\myvar{x}) - b(\myvar{x}^\prime)\|_{\infty}  \leq L_{b,\infty}\| \myvar{x} - \myvar{x}^\prime\|_{\infty}
\end{equation}
for any $\myvar{x},\myvar{x}^\prime $ in $\textup{Bound}(\myset{C})$. Let $b_i(\myvar{x})$ be the $i$th row of $b(\myvar{x})$. Thus, \eqref{eq:L_binf condition} is equivalent to $ | b_i(\myvar{x}) - b_i(\myvar{x}^\prime) | \leq L_{b,\infty}\| \myvar{x} - \myvar{x}^\prime\|_{\infty}, \forall i  = \{1,2\},$ for any $\myvar{x},\myvar{x}^\prime $ in $\textup{Bound}(\myset{C})$. Recall that
\begin{align}
    b_1(\myvar{x}) = 2\myvar{x}^\top Q \mathfrak{f}(\myvar{x}) + \myvar{x}^\top Q \myvar{x}-1 \\
    b_2(\myvar{x}) = 2\myvar{x}^\top Q \mathfrak{f}(\myvar{x}) - \myvar{x}^\top Q \myvar{x}+2
\end{align}
We have $\nabla b_1(\myvar{x}) = 2Q \mathfrak{f}(\myvar{x}) + 2 \frac{\partial \mathfrak{f}}{\partial \myvar{x}}(\myvar{x}) Q \myvar{x} + 2Q\myvar{x}$,  $\nabla b_2 (\myvar{x}) =  2Q \mathfrak{f}(\myvar{x}) + 2 \frac{\partial \mathfrak{f}}{\partial \myvar{x}}(\myvar{x}) Q \myvar{x} - 2Q\myvar{x}$, where $\frac{\partial \mathfrak{f}}{\partial \myvar{x}}(\myvar{x}) = \begin{psmallmatrix}
    1 & 1 \\
    -x_1 & 0
\end{psmallmatrix}
$. Following Mean Value Theorem, we know $| b_i(\myvar{x}) - b_i(\myvar{x}^\prime) | \leq \max_{\myvar{v}\in \textup{Bound}(\myset{C})} \|\nabla b_i(\myvar{v})\|_{\infty} \| \myvar{x} - \myvar{x}^\prime\|_{\infty}, \forall \myvar{x}, \myvar{x}^\prime \in\textup{Bound}(\myset{C}), \forall i \in \{1,2\}.$  Thus we choose $L_{b,\infty} \ge \max_{i = 1,2}  \max_{\myvar{v}\in \textup{Bound}(\myset{C})}( \|\nabla b_i(\myvar{v})\|_{\infty})  $. This leads to solve two quadratic programs and  we obtain $L_{b,\infty} = 13$.

Now we have all the necessary elements to execute \texttt{CompatibilityChecking}. Choose $\lambda = 0.25$. The Algorithm terminates after $3$ iterations and verifies the compatibility of the two CBFs. The execution process takes 169s and is shown in Fig. \ref{fig:example2}.

\end{example}

\begin{figure}[h!]
	\centering
	\includegraphics[width=0.8\linewidth]{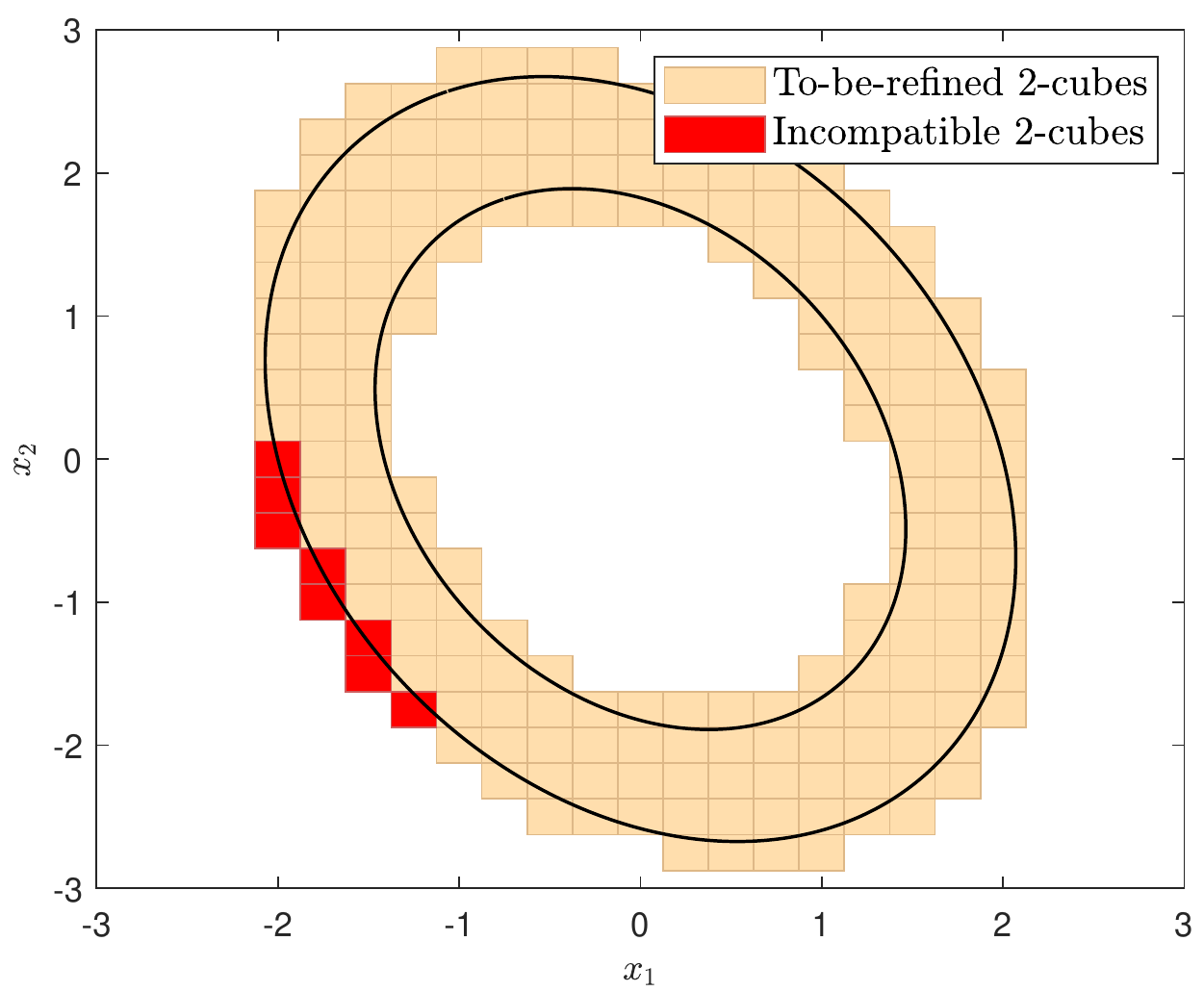}
    \caption{ Execution process of \texttt{CompatibilityChecking} in Example 3. The safety region is between the two ellipsoids. To-be-refined $2$-cubes: $2$-cubes that need further refinement (in yellow);  incompatible $2$-cubes: $2$-cubes whose center point is an incompatible state (in red). At the first iteration with the lattice size $r = 0.25$, Algorithm \ref{alg:compatibility checking} finds incompatible $2$-cubes and terminates, falsifying the CBF compatibility.  }  
    \label{fig:example3}	
	\end{figure}
	
\begin{example}
Now we consider the same scenario as in Example 2 but with a more stringent input set $\mathbb{U} =\{ (u_1,u_2): |u_1|\leq 2, |u_2| \leq 2\}$. This time, \texttt{CompatibilityChecking} gives a falsification on the compatibility. It finds an incompatible state $\myvar{x}_{in} = (-1.5,-1.25)$, at which point $h_1(\myvar{x}_{in}) =0.96, h_2(\myvar{x}_{in}) = 0.03, c(\myvar{x}_{in}) = -0.36$. The execution process takes 23s and is shown in Fig. \ref{fig:example3}.
\end{example}

\begin{example} 
In this example, a lower bound $\underline{r}$ of the size of the $2$-cubes is incorporated in Algorithm \ref{alg:compatibility checking} in a fashion described in Property 3) of Theorem \ref{thm:compatibility}.  The input set is again $\mathbb{U} =\{ (u_1,u_2): |u_1|\leq 3, |u_2| \leq 3\}$. In Example 2, we have already showed that the multiple CBFs are compatible. Now we would like to obtain an upper bound of the robustness level the multiple CBFs can attain. From Fig. \ref{fig:example2}, if we set $\underline{r} = 0.016$, then \texttt{CompatibilityChecking} takes 73s and terminates after $2$ iterations due to the lattice size decreases as the algorithm iterates and is smaller than $\underline{r}$ after the second iteration. Thus the termination belongs to \textbf{Case c} in the proof of Theorem \ref{thm:compatibility}. Based on \eqref{eq:eta_prime}, we thus know the multiple CBFs are at most robustly compatible with a robustness level $\eta = 0.6464 $. This is validated by, for example, considering that $c(-1.5,-1.25) = 0.5$.

\end{example}

\section{Conclusions}
In this work, we propose a grid sampling and refinement verification scheme for the compatibility checking of multiple control barrier functions for input constrained control systems. We provide both implementation details and theoretical guarantees for the verification problem. In particular, we show that if the algorithm terminates, it will give an exact answer to the compatibility problem. If the multiple CBFs are robustly compatible, then the algorithm is guaranteed to terminate in finite steps. If we also incorporate a lower bound on the size of the $n$-cubes, then we can also obtain the largest robustness level the multiple CBFs can possibly attain. We demonstrate the efficacy of the proposed algorithm in several scenarios that corroborates our theoretical results.

This work focuses on the compatibility checking of multiple CBFs. Future directions include how to determine the extended class $\mathcal{K}$ functions that mitigate the possible incompatibility  and/or increase the robustness level, and how to incorporate  the compatibility as a constraint with the online QP to ensure recursive feasibility.

\bibliographystyle{IEEEtran}
\bibliography{IEEEabrv,references}

\end{document}